\documentclass{article}

\PassOptionsToPackage{numbers, compress}{natbib}
\usepackage[preprint]{neurips_2026}

\usepackage[utf8]{inputenc} 
\usepackage[T1]{fontenc}    
\usepackage{url}            
\usepackage{booktabs}       
\usepackage{amsfonts}       
\usepackage{nicefrac}       
\usepackage{microtype}      
\usepackage{xcolor}         

\usepackage{graphicx}
\usepackage{algorithm}
\usepackage{algorithmic}
\usepackage{amsmath}
\usepackage{amsthm}
\newtheorem{proposition}{Proposition}   
\newtheorem{definition}{Definition}
\newtheorem{lemma}{Lemma}
\newtheorem{theorem}{Theorem}

\usepackage[utf8]{inputenc} 
\usepackage{booktabs}       
\usepackage{nicefrac}       
\usepackage{microtype}      
\usepackage{xcolor}         
\usepackage{mathtools}
\usepackage{multirow}
\usepackage{bm}

\usepackage{float}      
\usepackage{comment}    
\usepackage[table]{xcolor}
\usepackage{subcaption} 
\usepackage{hyperref}
\definecolor{academicblue}{RGB}{25,70,130}
\hypersetup{
    colorlinks=true,
    linkcolor=academicblue,
    citecolor=academicblue,
    urlcolor=black
}

\title{Vector Retrieval with Similarity and Diversity: How Hard Is It?}

\author{
Hang Gao$^{1}$ \qquad Dong Deng$^{1}$ \qquad Yongfeng Zhang$^{1}$\\
$^{1}$Rutgers University \\
\texttt{\{h.gao, dong.deng, yongfeng.zhang\}@rutgers.edu}
}

\begin{document}

\maketitle

\begin{abstract}
    Dense vector retrieval is an important building block of modern machine learning systems, underlying applications ranging from semantic search to retrieval-augmented generation and knowledge-intensive reasoning. Beyond retrieving items that are individually similar to a query, many applications require a set of results that is also diverse, complementary, and collectively informative. Balancing similarity and diversity is therefore central to effective retrieval, but remains challenging to optimize in a stable and theoretically grounded way. Maximal Marginal Relevance (MMR) is a widely adopted heuristic for this problem, yet its reliance on a manually tuned parameter leads to optimization fluctuations and unpredictable retrieval results. More broadly, existing methods provide limited theoretical insight into how similarity and diversity interact in dense vector spaces, leaving the joint optimization problem insufficiently understood. To address these challenges, this paper introduces a novel approach that characterizes both constraints simultaneously by maximizing the similarity between the query vector and the sum of the selected candidate vectors. We formally define this optimization problem, Vector Retrieval with Similarity and Diversity (VRSD), and prove that it is NP-complete, establishing a rigorous theoretical bound on the inherent difficulty of this dual-objective retrieval. Subsequently, we present a parameter-free heuristic algorithm to solve VRSD. Extensive evaluations on multiple datasets, incorporating both objective geometric metrics and LLM-simulated subjective assessments, demonstrate that our VRSD heuristic consistently outperforms established baselines, including MMR and Determinantal Point Processes (k-DPP).
\end{abstract}

\section{Introduction}
Dense vector retrieval has emerged as a foundational mechanism in Natural Language Processing (NLP), acting as a critical driver for Large Language Models (LLMs). Its application has driven significant performance gains across diverse domains, most notably in open-domain question answering \citep{chen2017reading,guu2020retrieval,khattab2020colbert,ding2020rocketqa,izacard2020leveraging}. Furthermore, the advent of Retrieval-Augmented Generation (RAG) \citep{lewis2020retrieval, karpukhin2020dense, mao2020generation, trivedi2023interleaving, edge2024local} has underscored the indispensability of accurate retrieval in knowledge-intensive tasks to effectively ground and improve generation quality.

However, optimizing for pure relevance is often insufficient. Incorporating diversity is essential to maximize informational coverage and mitigate semantic redundancy within the retrieved context window \citep{santos2015search}. Rather than being mutually exclusive, similarity and diversity share a conceptual parallel with the exploration–exploitation trade-off in reinforcement learning \citep{sutton2018reinforcement}. Consequently, engineering an optimal balance between targeting exact matches and exploring diverse facets remains a central challenge in modern retrieval systems.

Maximal Marginal Relevance (MMR) \citep{carbonell1998use} is the standard approach for this balance, widely deployed in vector databases (e.g., Qdrant \citep{QdrantDocs}, Pinecone \citep{pinecone2023}, Milvus \citep{wang2021milvus}, Weaviate \citep{de2021weaviate}) and frameworks such as LangChain \citep{chase2022langchain} and LlamaIndex \citep{Liu_LlamaIndex_2022}. MMR uses a parameter, \(\lambda\), to weigh relevance and diversity. However, the optimal  varies across scenarios and cannot be known a priori \citep{deselaers2009jointly, ye2023complementary, wang2025diversityenhancesllmsperformance}. Furthermore, recent work \citep{rubin2021learning,wang2023learning} suggests framing retrieval as a combinatorial optimization problem rather than independent selections, as examples in the context influence each other. Inspired by this, we propose using the sum vector to characterize both similarity and diversity in vector retrieval. In essence, this involves maximizing the similarity between the sum vector of the selected vectors and the query vector, and maximizing the similarity of the sum vector to the query vector imposes a similarity constraint. At the same time, from a geometric perspective, the requirement for the sum vector to be similar to the query vector means that the selected vectors approach the query vector from different directions, thus imposing a diversity constraint. 

In summary, our main contributions are as follows.
\begin{itemize}
\item Novel Unified Framework: We propose VRSD, a parameter-free vector retrieval approach that naturally unifies similarity and diversity constraints by maximizing the alignment between the query and the sum of the retrieved vectors.

\item Theoretical Complexity Bound: We formally define the VRSD optimization problem and theoretically prove it is NP-complete via reduction from the subset sum problem, highlighting the inherent difficulty of simultaneously achieving relevance and diversity.

\item Efficient Heuristic and Empirical Validation: We introduce an efficient heuristic algorithm and demonstrate that it consistently outperforms established baselines (MMR and k-DPP) on multiple scientific QA datasets, validated through both objective geometric metrics and LLM-simulated human judgments.
\end{itemize}

\section{Theoretical analysis of MMR}
\subsection{Limitations of MMR}

MMR \citep{carbonell1998use} addresses the balance between relevance and diversity by employing "marginal relevance" as an evaluation metric. The metric is defined as a linear combination of independently measured relevance and novelty, formulated as Eq.\ref{eq:1}:
\begin{equation}
\label{eq:1}
\begin{split}
\text{MMR}=\arg\max_{d_i \in R \setminus S} [ \lambda \cdot \text{Sim}_1(d_i, q) - (1 - \lambda) \cdot \max_{d_j \in S} \text{Sim}_2(d_i, d_j) ].
\end{split}
\end{equation}

Here, $R$ denotes the candidate set and $S$ is the subset of already selected elements. The expression $R \setminus S$ represents the set difference, that is, the set of unselected elements in $R$. $\text{Sim}_1$ and $\text{Sim}_2$ are similarity metrics. The parameter $\lambda \in [0,1]$ is a weighting factor. When $\lambda = 1$, MMR incrementally computes a standard relevance-ranked list. In contrast, when $\lambda = 0$, it produces a ranking that maximizes the diversity between elements in $S$.

\begin{figure*}
    \centering
    \includegraphics[width=1 \textwidth]{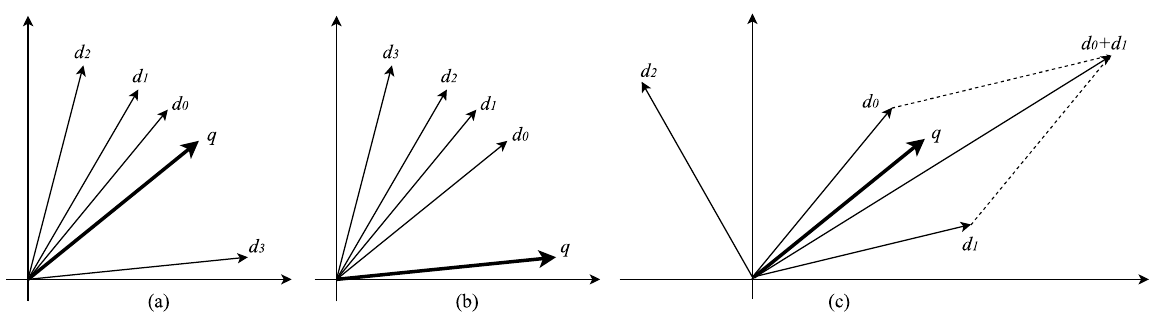}
    \caption{An analysis of the Maximal Marginal Relevance. (a) The candidate vectors are located on different sides of the query vector. (b) The candidate vectors are located on the same side of the query vector. (c) An example for MMR and the sum vector.
    }
    \label{fig:fig1}
    \vspace{-4mm}
\end{figure*}

The challenge lies in selecting an appropriate \(\lambda\) to achieve the desired balance between relevance and diversity, particularly in high-dimensional vector spaces where the impact of varying \(\lambda\) is less predictable. This variability in \( \lambda \) leads to fluctuations in retrieval results, resulting in unpredictable consequences, which can be illustrated by a simple example. Typically, \( \lambda \) is pre-set at a value of 0.5 in many MMR implementations, a choice that comes from the algorithm's foundational design. 
As illustrated in Figure \ref{fig:fig1}(a), consider \( q \) as the query vector and \( d_0 \) to \( d_3 \) as candidate vectors that exceed the relevance threshold, collectively represented as \( R = \{d_0, d_1, d_2, d_3\} \), with \( S \) initially empty. Using the MMR algorithm, \( d_0 \) is first selected due to its highest relevance to \( q \), determined using cosine similarity as a measure. Subsequently, \( d_3 \) is chosen over \( d_1 \), despite \( d_1 \) having a smaller angle with \( q \) and thus greater direct relevance. The selection of \( d_3 \) is influenced by the fact that the cumulative relevance between \( d_1 \) and \( d_0 \) significantly exceeds that between \( d_3 \) and \( d_0 \), resulting in a higher MMR value for \( d_3 \) according to the formula. 

However, as depicted in Figure \ref{fig:fig1}(b), with \( q \) serving as the query vector and \( R = \{d_0, d_1, d_2, d_3\} \) representing the initial set of candidate vectors, \( d_0 \) is selected first due to its maximal relevance to \( q \). The selection using the MMR proceeds as follows: with \( \lambda = 0.5 \), \( S = \{d_0\} \), and \( R\setminus S = \{d_1, d_2, d_3\} \), the formula can be formed as Eq.\ref{eq:2}:
\begin{equation}
\label{eq:2}
\mathrm{MMR} = \arg\max_{i \in \{1,2,3\}}\Bigl[ 0.5 \cdot (\mathrm{Sim}_1(d_i, q) - \mathrm{Sim}_2(d_i, d_0)) \Bigr].
\end{equation}

Given that \( d_0 \), \( d_1 \), \( d_2 \), and \( d_3 \) are positioned on the same side relative to \( q \), and assuming that both \( \text{Sim}_1 \) and \( \text{Sim}_2 \) denote cosine similarity, let \( \theta \) represent the angle between \( d_0 \) and \( q \), and \( x \) denote the angle between \( d_i \) (i.e., \( d_1, d_2, d_3 \)) and \( d_0 \). Thus, we get the Eq.\ref{eq:3}:
\begin{equation}
\label{eq:3}
\hspace*{-3mm}
\text{MMR} = \arg\max_{\mathclap{i=1,2,3}} \big[ 0.5 \cdot (\cos(d_i, q) - \cos(d_i, d_0)) \big]
= \arg\max_{\mathclap{i=1,2,3}} \big[ 0.5 \cdot (\cos(x+\theta) - \cos(x)) \big]
\end{equation}

The function \( f(x) = \cos(x+\theta) - \cos(x) \), with its derivative \( f'(x) = -\sin(x+\theta) + \sin(x) \), assumes that \( x \) and \( x+\theta \) lie within \((0, \pi/2)\). Consequently, \( f'(x) < 0 \), indicating that for vectors on the same side of \( q \), their MMR values decrease as the angle with \( q \) increases. Thus, following the selection of \( d_0 \), the subsequent choices are \( d_1 \), then \( d_2 \), and so on. This sequence suggests that relevance predominantly influences the selection.

The real challenge in vector retrieval emerges when \( \lambda \neq 0.5 \). The selection among the candidate vectors \( d_1 \), \( d_2 \), and \( d_3 \) is critically dependent on both \( \lambda \) and \( \theta \), which complicates the determination of the most appropriate candidate. As shown in Figure \ref{fig:fig1}(c), let $q$ be the query vector and $d_0$, $d_1$, and $d_2$ be the candidate vectors. If the parameter $\lambda$ in MMR is set to a small value, indicating a stronger preference for diversity, the selection process can proceed as follows: $d_0$ is selected first due to its highest similarity to $q$. Subsequently, $d_2$ is chosen over $d_1$ because it is farther from $d_0$, thus exhibiting greater diversity. However, this result is often suboptimal in real-world applications, where the desired diversity is typically based on similarity. That is, diversity should be constrained by relevance to the query. Selecting candidates that are too dissimilar to the query undermines the purpose of retrieval and is of limited practical value.

\subsection{Sum Vector for Retrieval}
To address the challenges posed by MMR, we propose using the sum vector to simultaneously capture both similarity and diversity. Specifically, during the selection process from candidate vectors, we compute the cosine similarity between the query vector and the sum of the selected vectors, in order to maximize this similarity. We again illustrate this with Figure\ref{fig:fig1}(c). As before, $q$ is the query vector, and $d_0$, $d_1$, and $d_2$ are the candidate vectors. Using the sum vector approach, $d_0$ is selected first due to its highest similarity to $q$. Then, based on the rule of vector addition, the sum of $d_0$ and $d_1$ is evidently more similar to $q$ than the sum of $d_0$ and $d_2$, and therefore $d_1$ is selected next - unlike in MMR, where $d_2$ could be chosen.

The key insight lies in the properties of vectors and vector addition. Geometrically, as illustrated in Figure \ref{fig:fig1}(c), the sum of two vectors lies between them. Furthermore, we can rigorously prove this proposition using algebraic methods. 
\begin{proposition}
\label{proposition1}
In an $n$-dimensional vector space, the sum of two vectors necessarily lies between the two original vectors. That is, let $u, v \in \mathbb{R}^n$, and $s=u+v$, so $s$ must be between $u$ and $v$.
\end{proposition}
The main idea of the proof is that $u$, $v$, and $s$ lie in the same hyperplane, and then it can be proved that the angle between $s$ and $u$ is less than the angle between $u$ and $v$, and the angle between $s$ and $v$ is also less than the angle between $u$ and $v$, therefore, $s$ must be between $u$ and $v$. A detailed formal proof of Proposition \ref{proposition1} is provided in the Appendix \ref{proof}.

Therefore, from the algorithmic perspective, consider a query vector $q$ and a set of candidate vectors $d_0, d_1, \ldots, d_n$, sorted by their similarity to $q$. When the goal is to select a subset of vectors such that the sum vector is maximally similar to $q$, the algorithm first selects $d_0$. For the second selection, the next best candidate $d_x$ is chosen such that the new sum $s = d_0 + d_x$ is as close as possible to $q$. Based on Proposition \ref{proposition1}, this sum vector $s$ lies between $d_0$ and $d_x$, which implies that $d_0$ and $d_x$ are positioned on either side of $s$, and hence $d_0$ and $d_x$ are positioned on either side of $q$ when optimized. This introduces an implicit diversity constraint, since selected vectors are encouraged to approach $q$ from different directions.

Repeating this process means that each new vector is selected so that its addition pushes the sum vector closer to the query while maintaining angular spread, achieving a balance between similarity and diversity.

In fact, this property of vector addition has been observed in many influential works. For example, in one of the most important studies on word embeddings, word2vec \citep{Mikolov2013EfficientEO}, the famous example "king - man + woman = queen" demonstrates how semantic relationships can be captured through simple vector arithmetic. Similarly, our use of the sum vector captures the joint semantics of selected vectors. By evaluating the similarity between the sum vector and the query, we explicitly enforce relevance, while the geometric nature of vector addition imposes an adaptive diversity constraint.

\section{Vectors retrieval with similarity and diversity}
\subsection{Problem Definition and Complexity}
To address the problem of selecting a subset of vectors from a set of candidate vectors that satisfy both similarity and diversity requirements, we refer to the MMR algorithm and several related algorithms, typically considering the following premise. The candidate vectors are identified from the entire set of vectors (\(size = N\)) using cosine similarity metrics, resulting in a subset of vectors (\(size = n\)). Consequently, this set of candidate vectors inherently exhibits a relatively high degree of similarity to the query vector. 

As mentioned above, while algorithms like MMR are widely applied in practice, these studies often lack a robust and reliable theoretical model. In other words, many approaches employ heuristic strategies or machine learning methods to arrive at a solution without providing a rigorous formal description and analysis of similarity and diversity from a theoretical perspective. Therefore, based on the aforementioned premises, we propose using the sum vector to characterize both similarity and diversity in vector retrieval. The definition of the sum vector is as follows:

\begin{definition}
\label{sum vector}
The Sum Vector: Given \(k\) vectors \(d_0, d_1, ..., d_{k-1}\), the sum vector \(d\) is the sum of these \(k\) vectors.
\end{definition}

Next, we define the vector retrieval problem as follows:

\begin{definition}
\label{The VRSD problem}
The problem of Vector Retrieval with Similarity and Diversity (VRSD): Given a query vector \( q \) and a set of candidate vectors \( R = \{d_0, d_1, ..., d_{n-1}\} \), how to select \( k \) vectors (\(d'_0, d'_1, ..., d'_{k-1}\)) from \( R \) so that the cosine similarity between the sum vector \( d = d'_0 + d'_1 + ... + d'_{k-1}\) and \( q \) is maximized?
\end{definition}

On the surface, requiring the sum of the selected vectors to maximize cosine similarity with the query vector imposes a similarity constraint, and the sum vector also represents joint semantics, similar to the well-known example ("king - man + woman = queen") in word2vec \citep{Mikolov2013EfficientEO}. Furthermore, according to Proposition \ref{proposition1}, requiring the sum vector to be similar to the query vector implicitly imposes a diversity constraint. Therefore, the VRSD problem we define effectively captures both similarity and diversity constraints in vector retrieval. However, upon further examination of the above problem, we find that it is an NP-hard problem. In the following, we provide a theoretical proof. To this end, we first give the corresponding decision problem and prove its NP-completeness by reducing the subset sum problem to the VRSD decision problem.

\begin{definition}
\label{the decision problem}
The decision problem of vector retrieval: Given a set of candidate vectors \(R\) and a query vector \(q\), can \(k\) vectors be selected from \(R\) such that the cosine similarity between the sum vector of these \(k\) vectors and the query vector \(q\) is equal to 1? We denote instances of this vector retrieval problem by \((R, q, k)\).
\end{definition}

Next, we will prove that this decision problem is NP-complete. For the sake of concise proof, we further restrict the components of vectors to integers. The proof strategy is to reduce the subset sum problem \cite{cormen2009introduction} to this decision problem.

\begin{definition}
\label{the subset problem}
The subset sum problem: Given an integer set \(T\) and another integer \(t\), does there exist a nonempty subset whose sum of elements is equal to \(t\)? We denote instances of the subset sum problem as \((T, t)\).
\end{definition}

For convenience of proof, we also need to define a modified version of the subset sum problem, called the \(k\)-subset sum problem.

\begin{definition}
\label{k sum problem}
\(k\)-subset sum problem: Given an integer set \(T\) and another integer \(t\), does there exist a non-empty subset of size \(k\) (i.e., the cardinality of the subset is \(k\)), whose sum of elements equals \(t\)? We denote instances of the \(k\)-subset sum problem as \((T, t, k)\).
\end{definition}

\begin{lemma}
The \(k\)-subset sum problem is NP-complete.
\end{lemma}

\begin{proof} 
We reduce the subset sum problem(Def.\ref{the subset problem}) to the \(k\)-subset sum problem(Def.\ref{k sum problem}) .

1. Clearly, the \(k\)-subset sum problem is polynomial-time verifiable.

2. Reducing the subset sum problem to the \(k\)-subset sum problem.

For any instance of the subset sum problem \((T, t)\), where \(T=\{t_1, t_2, ..., t_n\}, |T|=n\),
we can transform it into an instance of the \(k\)-subset sum problem \((T',t',k)\), where \(T'=T \cup \{0, 0, ..., 0\}\) (that is, add $n$ zero elements to $T$), \(t'=t\), \(k=n\).

It is obvious that if \((T',t',k)\) has a yes answer, then the answer to \((T, t)\) is yes. If the answer to \((T',t',k)\) is no, then the answer to \((T, t)\) is no. Therefore, if the \(k\)-subset sum can be solved in polynomial time, then the subset sum can also be solved in polynomial time. Hence, the \(k\)-subset sum problem is NP-complete.
\end{proof} 
Now it is time to prove the NP-completeness of the decision problem of vector retrieval.
\begin{theorem}
\label{Th_vec_ret}
The decision problem of vector retrieval is NP-complete.
\end{theorem}

\begin{proof}
We reduce the \(k\)-subset sum problem(Def.\ref{k sum problem}) to the decision problem of vector retrieval(Def.\ref{the decision problem}).

1. The answer to vector retrieval is polynomial-time verifiable. If the answer provides \(k\) vectors, we can simply add these \(k\) vectors and then calculate whether the cosine similarity between the sum vector and the query vector \(q\) equals 1. This verification can be performed in polynomial time.

2. Reducing the \(k\)-subset sum problem to the decision problem of vector retrieval.

For any instance of the \(k\)-subset sum problem \((T, t, k)\), let \(T = \{t_1, t_2, \ldots, t_n\}\). We construct the set of vectors \(R\) and the query vector \(q\) as Eq.\ref{eq:4}:
\begin{equation}
\label{eq:4}
R = \{[t_1, 1], [t_2, 1], \ldots, [t_n, 1]\}, q = [t, k]\
\end{equation}

The decision problem of vector retrieval \((R, q, k)\) asks whether there exist \(k\) vectors such that the sum vector (denoted as \(d\)) of these vectors and the query vector \(q\) have a cosine similarity of 1. According to the definition of cosine similarity, \(cos\_similarity = \frac{d \cdot q}{|d| \cdot |q|}\). The cosine similarity between \(d\) and \(q\) equals 1 if and only if \(d = \alpha q\), where \(\alpha\) is a scalar.
Therefore, if vector retrieval provides an affirmative answer \(d = \alpha q\), we can get the Eq.\ref{eq:5}, 
\begin{equation}
\label{eq:5}
\begin{split}
d = [t'_1, 1] + [t'_2, 1] + ... + [t'_k, 1] = \alpha [t, k]
\Rightarrow [(t'_1 + ... + t'_k), k] = \alpha [t, k].
\end{split}
\end{equation}

\( [t'_1, 1] \ldots [t'_k, 1] \) are the selected \( k \) vectors. This implies that \(\alpha = 1\) and \(t'_1 + \ldots + t'_k = t\). Thus, this provides an affirmative answer to the \(k\)-subset sum problem instance \((T, t, k)\). Conversely, if vector retrieval provides a negative answer, then a negative answer to the \(k\)-subset sum problem can also be obtained. The above reduction process can be clearly completed in polynomial time. \textbf{Therefore, the decision problem of vector retrieval is NP complete.}
\end{proof}

\subsection{Infeasibility of Dynamic Programming}

It is well known that although the subset sum problem is NP-complete, there exists a simple pseudo-polynomial-time dynamic programming algorithm for solving it. Specifically, for the subset sum problem as defined in Def.\ref{the subset problem}, the dynamic programming recurrence is given by:
\[
DP[i, t] = \left\{
\begin{array}{l}
DP[i-1, t], t_i \text{is not selected} \\ 
DP[i-1, t - t_i], t_i \text{is selected}
\end{array} \right.
\]

Although this dynamic programming approach is essentially pseudo-polynomial, it can be effectively applied to solve the subset sum problem in certain scenarios. This raises the question: Does there exist a similarly simple dynamic programming algorithm for the VRSD problem?

Unfortunately, the answer is negative. As noted in the proof of Theorem \ref{Th_vec_ret}, particularly in Eq.\ref{eq:5}, to get the cosine similarity between the sum vector and the query vector equal to 1, we must hold $d = \alpha q$ rather than $d = q$. Consequently, if we attempt to apply a dynamic programming approach similar to that of the subset sum problem, we face the difficulty that the final sum vector cannot be determined in advance, since the scalar $\alpha$ in $d = \alpha q$ is undetermined. Therefore, a dynamic programming algorithm analogous to that of the subset sum problem does not exist for the VRSD problem.

\subsection{Heuristic algorithm for vector retrieval}
Since the vector retrieval problem \((R, q, k)\) is an NP-hard problem, it is necessary to use heuristic methods to derive feasible solutions. Specifically, given a set of candidate vectors with high similarity, the objective is to select \( k \) vectors that maximize the cosine similarity between the sum vector of the \(k\) selected vectors and the query vector. We propose a heuristic algorithm also named as Vector Retrieval with Similarity and Diversity (VRSD). In each iteration, it chooses the vector that maximizes the cosine similarity between the sum of all selected vectors and the query vector, continuing this process until the \(k\) vectors are chosen. More details about the VRSD algorithm can be found in Algorithm \ref{algo_vec_ret}. 

\begin{algorithm}[htb]
\caption{Vector Retrieval with Similarity and Diversity (VRSD)}
\label{algo_vec_ret}
\textbf{Input}: Candidate vector set \( R = \{d_0, d_1, \ldots, d_{n-1}\} \), query vector \( q \).

\textbf{Output}: \( k \) vectors such that the cosine similarity between the sum vector of these \( k \) vectors and \( q \) is maximized.
\begin{algorithmic}[1] 
\STATE \( S = \{ \} \)
\FOR{\( i = 1 \) to \( k \)}
    \STATE \( s = \sum S \) \COMMENT{Sum of all vectors in \( S \)}
    \STATE \( \text{max\_cos} = -1 \)
    \STATE \( p = \text{null} \) \COMMENT{Initialize \( p \)}
    \FOR{\( v \) in \( R \)}
        \STATE \( t = s + v \) \COMMENT{Temporary vector}
        \IF{\( \cos(t, q) > \text{max\_cos} \)}
            \STATE \( \text{max\_cos} = \cos(t, q) \)
            \STATE \( p = v \)
        \ENDIF
    \ENDFOR
    \STATE \( S = S \cup \{p\} \) \COMMENT{Add \( p \) to the set \( S \)}
    \STATE \( R = R \setminus \{p\} \) \COMMENT{Remove \( p \) from \( R \)}
\ENDFOR
\STATE return \(S\)
\end{algorithmic}
\end{algorithm}
\vspace{-3mm}

\subsection{Time complexity analysis of VRSD}
As depicted in Algorithm \ref{algo_vec_ret}, the time complexity of the VRSD algorithm is \( k \times |R| = k \times n \), which accounts for the initial step of selecting \( n \) candidate vectors from the entire set of vectors (size = \( N \)) based on similarity. Given that \( N \gg n > k \), the computational load of subsequent steps in Algorithm \ref{algo_vec_ret} is minimal compared to the MMR algorithm, which also selects \( k \) vectors from \( |R| \) candidates, requires two iterations of maximum calculations as depicted in Eq.\ref{eq:1}—once for each candidate vector against the query vector and once against the set of vectors already selected \( |S| \). Thus, the complexity of MMR becomes \( k \times |R| \times |S| = k \times |R|^2 = k \times n^2\), indicating a marginally higher computational demand compared to VRSD.

\section{Experiments}
\subsection{Baselines and Datasets}
We benchmark VRSD against two established baselines: \textbf{MMR}, a standard greedy heuristic, and \textbf{k-DPP} \citep{kulesza2011kdpp}, a probabilistic approach modeling item repulsion (implemented via Fast-MAP \citep{chen2018fastmap}). These represent complementary strategies widely used in summarization \citep{sharghi2018seqdppimprove}. We employ the \texttt{all-mpnet-base-v2} encoder \citep{song2020mpnet}. Furthermore, ablation studies are performed using two alternative embedding models, as detailed in the Appendix \ref{comprehensive_AB}. Evaluations are conducted on three scientific QA datasets: \textbf{ARC-DA} \cite{bhakthavatsalam2021think}, \textbf{OpenBookQA} \cite{mihaylov2018openbookqa}, and \textbf{SciQ} \cite{welbl2017sciq}, using both objective metrics and LLM-simulated subjective assessments. More details can be found in Appendix \ref{app:setup_details}.

\subsection{Objective Evaluation}
The objective evaluation compares VRSD and the baselines from two complementary perspectives: complementarity-aware similarity and pairwise diversity. The first metric evaluates whether the retrieved set is collectively aligned with the query, rather than whether each retrieved item is independently similar to the query. The second metric measures the average redundancy among retrieved items. Together, these metrics allow us to examine whether a method can retrieve results that are both query-aligned and non-redundant.

\subsubsection{Complementarity-aware Similarity}
\begin{table}[t]
\centering
\caption{Similarity and Diversity scores on ARC, OpenBookQA, and SciQ for varying k.}
\resizebox{\textwidth}{!}{%
\begin{tabular}{l|l|cc|cc|cc}
\toprule
$k$ & \textbf{Algorithms} & \multicolumn{2}{c|}{\textbf{ARC}} & \multicolumn{2}{c|}{\textbf{OpenBookQA}} & \multicolumn{2}{c}{\textbf{SciQ}} \\
 &  & Sim. Mean & Div. Mean & Sim. Mean & Div. Mean & Sim. Mean & Div. Mean \\
\midrule
\multirow{9}{*}{6}
& MMR($\lambda$=0.2) & 0.6797 & 0.2766 & 0.6592 & 0.2737 & 0.6609 & 0.2363 \\
& MMR($\lambda$=0.3) & 0.6859 & 0.2789 & 0.6652 & 0.2785 & 0.6649 & 0.2388 \\
& MMR($\lambda$=0.4) & 0.6913 & 0.2844 & 0.6739 & 0.2825 & 0.6701 & 0.2432 \\
& \textbf{MMR($\lambda$=0.5)} & \textbf{0.6987} & \textbf{0.2988} & \textbf{0.6789} & \textbf{0.2937} & \textbf{0.6789} & \textbf{0.2549} \\
& \textbf{MMR($\lambda$=0.6)} & \textbf{0.7047} & \textbf{0.3178} & \textbf{0.6875} & \textbf{0.3148} & \textbf{0.6880} & \textbf{0.2747} \\
& MMR($\lambda$=0.7) & 0.7065 & 0.3396 & 0.6890 & 0.3427 & 0.6890 & 0.3036 \\
& MMR($\lambda$=0.8) & 0.7022 & 0.3631 & 0.6860 & 0.3636 & 0.6843 & 0.3314 \\
& MMR($\lambda$=0.9) & 0.6928 & 0.3914 & 0.6779 & 0.3868 & 0.6732 & 0.3603 \\
& k-DPP                & 0.7081 & 0.3382 & 0.6899 & 0.3455 & 0.6872 & 0.3198 \\
& \textbf{VRSD}            & \textbf{0.7161} & \textbf{0.3109} & \textbf{0.6996} & \textbf{0.3088} & \textbf{0.6987} & \textbf{0.2721} \\
\midrule
\multirow{9}{*}{12}
& MMR($\lambda$=0.2) & 0.6780 & 0.2399 & 0.6656 & 0.2390 & 0.6695 & 0.2020 \\
& MMR($\lambda$=0.3) & 0.6869 & 0.2442 & 0.6721 & 0.2437 & 0.6759 & 0.2055 \\
& MMR($\lambda$=0.4) & 0.6981 & 0.2535 & 0.6820 & 0.2500 & 0.6860 & 0.2129 \\
& \textbf{MMR($\lambda$=0.5)} & \textbf{0.7073} & \textbf{0.2661} & \textbf{0.6919} & \textbf{0.2629} & \textbf{0.6968} & \textbf{0.2268} \\
& \textbf{MMR($\lambda$=0.6)} & \textbf{0.7147} & \textbf{0.2857} & \textbf{0.7013} & \textbf{0.2847} & \textbf{0.7042} & \textbf{0.2453} \\
& MMR($\lambda$=0.7) & 0.7143 & 0.3122 & 0.7005 & 0.3074 & 0.7041 & 0.2738 \\
& MMR($\lambda$=0.8) & 0.7100 & 0.3323 & 0.6940 & 0.3291 & 0.6978 & 0.2977 \\
& MMR($\lambda$=0.9) & 0.7022 & 0.3506 & 0.6861 & 0.3456 & 0.6881 & 0.3169 \\
& k-DPP                & 0.7155 & 0.3136 & 0.7000 & 0.3162 & 0.7016 & 0.2866 \\
& \textbf{VRSD}             & \textbf{0.7332} & \textbf{0.2695} & \textbf{0.7177} & \textbf{0.2668} & \textbf{0.7235} & \textbf{0.2330} \\
\midrule
\multirow{9}{*}{18}
& MMR($\lambda$=0.2) & 0.6712 & 0.2213 & 0.6601 & 0.2210 & 0.6649 & 0.1839 \\
& MMR($\lambda$=0.3) & 0.6789 & 0.2253 & 0.6684 & 0.2254 & 0.6718 & 0.1884 \\
& MMR($\lambda$=0.4) & 0.6918 & 0.2332 & 0.6805 & 0.2352 & 0.6862 & 0.1966 \\
& \textbf{MMR($\lambda$=0.5)} & \textbf{0.7039} & \textbf{0.2475} & \textbf{0.6921} & \textbf{0.2496} & \textbf{0.6993} & \textbf{0.2098} \\
& \textbf{MMR($\lambda$=0.6)} & \textbf{0.7117} & \textbf{0.2674} & \textbf{0.6988} & \textbf{0.2697} & \textbf{0.7071} & \textbf{0.2293} \\
& MMR($\lambda$=0.7) & 0.7105 & 0.2900 & 0.6972 & 0.2900 & 0.7062 & 0.2513 \\
& MMR($\lambda$=0.8) & 0.7035 & 0.3143 & 0.6897 & 0.3108 & 0.6988 & 0.2732 \\
& MMR($\lambda$=0.9) & 0.6944 & 0.3326 & 0.6826 & 0.3246 & 0.6868 & 0.2935 \\
& k-DPP                & 0.7112 & 0.2933 & 0.6970 & 0.2963 & 0.7021 & 0.2657 \\
& \textbf{VRSD}             & \textbf{0.7344} & \textbf{0.2454} & \textbf{0.7187} & \textbf{0.2459} & \textbf{0.7290} & \textbf{0.2104} \\
\bottomrule
\end{tabular}
}
\label{tab:sim-div-combined}
\vspace{-5mm}
\end{table}

We first evaluate complementarity-aware similarity using sum-vector alignment. Given a retrieved set $\{d_1, d_2, \ldots, d_k\}$, we construct a set-level representation by summing the selected vectors:
\[
d_{\mathrm{sum}} = d_1 + d_2 + \cdots + d_k.
\]
We then compute the cosine similarity between $d_{\mathrm{sum}}$ and the query vector $q$:
\begin{equation}
\label{eq:7}
\text{Similarity} = \cos(d_{\mathrm{sum}}, q).
\end{equation}

Unlike conventional query-item similarity, which evaluates each retrieved item independently, this similarity evaluates the retrieved set as a whole. A high score indicates that the collective semantics of the selected vectors are well aligned with the query. At the same time, because vector addition allows different vectors to contribute from different directions, this metric also captures relevance-aware complementarity: selected items can be individually distinct while jointly forming a semantic direction close to the query. 

For this experiment, 20\% of each data set was used as a set of retrieval queries. 
For MMR, we set $\lambda$ values ranging from 0.2 to 0.9 to thoroughly and fairly evaluate the performance of the MMR algorithm. In each retrieval task, once a query is submitted to VRSD, MMR or k-DPP, $k$ vectors are returned. We sum these $k$ vectors to compute the cosine similarity between the sum vector and the query vector. The mean cosine similarity for all test queries is reported as "Sim. Mean" in Table \ref{tab:sim-div-combined}.

From the "Similarity Mean" columns in Table \ref{tab:sim-div-combined}, we observe that regardless of the value of $\lambda$ in MMR, as expected, VRSD achieves the highest score across datasets than MMR and k-DPP, confirming that the proposed greedy heuristic effectively optimizes the sum-vector objective. More importantly, the following pairwise diversity results show that this objective alignment does not collapse into redundant retrieval, but instead preserves strong diversity under a relevance-aware constraint.

\subsubsection{Diversity}

While similarity measures set-level query alignment, it does not by itself quantify the degree of redundancy among selected items. We therefore additionally adopt the commonly used metric in recommendation and RAG, pairwise similarity \citep{ziegler2005improving, jesse2023intra}. Formally, 
pairwise similarity is defined as the average similarity between all distinct pairs of elements in a given set. It is often used to measure the diversity of a set of objects.

The \textbf{Diversity} metric of \(\{d_1, d_2, \ldots, d_k\} \) is defined as:
\begin{equation}
\label{eq:7}
\text{Diversity} = \frac{2}{k(k - 1)} \sum_{i=1}^{k} \sum_{j=i+1}^{k} \cos(d_i, d_j)
\end{equation}
A lower Diversity score implies a higher degree of diversity in the set. We report the average diversity scores for the 20\% test queries in the "Div. Mean" column of Table \ref{tab:sim-div-combined}. As shown in Table \ref{tab:sim-div-combined}, with increasing $\lambda$, the diversity scores of MMR gradually increase, which means that its diversity performance gradually decreases. This aligns with the MMR objective: The larger $\lambda$ emphasizes similarity over diversity. More notably, when comparing diversity scores of VRSD and MMR in Table \ref{tab:sim-div-combined}, we observe the following. 

\begin{itemize}
\item When $\lambda < 0.4$, MMR exhibits better diversity than VRSD, as it strongly prioritizes diversity.
\item When $\lambda \in (0.4, 0.6)$, VRSD and MMR have roughly comparable diversity.
\item When $\lambda > 0.6$, VRSD begins to outperform MMR in diversity, which is expected since the emphasis of MMR on diversity decreases as $\lambda$ increases.
\end{itemize}

We highlight MMR ($\lambda \in \{0.5, 0.6\}$) and VRSD in Table \ref{tab:sim-div-combined} to facilitate comparison. Results indicate that VRSD offers robust diversity: it matches MMR in the critical $0.5\text{-}0.6$ range and surpasses it as $\lambda$ increases. Furthermore, VRSD consistently outperforms k-DPP on both metrics across all datasets and $k$ values. Comparing the baselines, k-DPP behaves similarly to MMR with $\lambda \in [0.5, 0.7]$, typically achieving slightly higher similarity but lower diversity than MMR in this range.

The above comparison between VRSD, MMR and k-DPP reveals that VRSD not only maximizes the similarity between the sum vector and the query vector, but also adaptively achieves diversity.

Considering that in most practical applications the desired diversity is similarity aware, diversity without any similarity constraint is often of limited significance. Hence, a preliminary conclusion can be drawn: VRSD outperforms MMR and k-DPP when considering both similarity and diversity in a unified framework.

\subsection{Subjective Evaluation}
\begin{figure}[t]
    \centering
    \includegraphics[width=\textwidth]{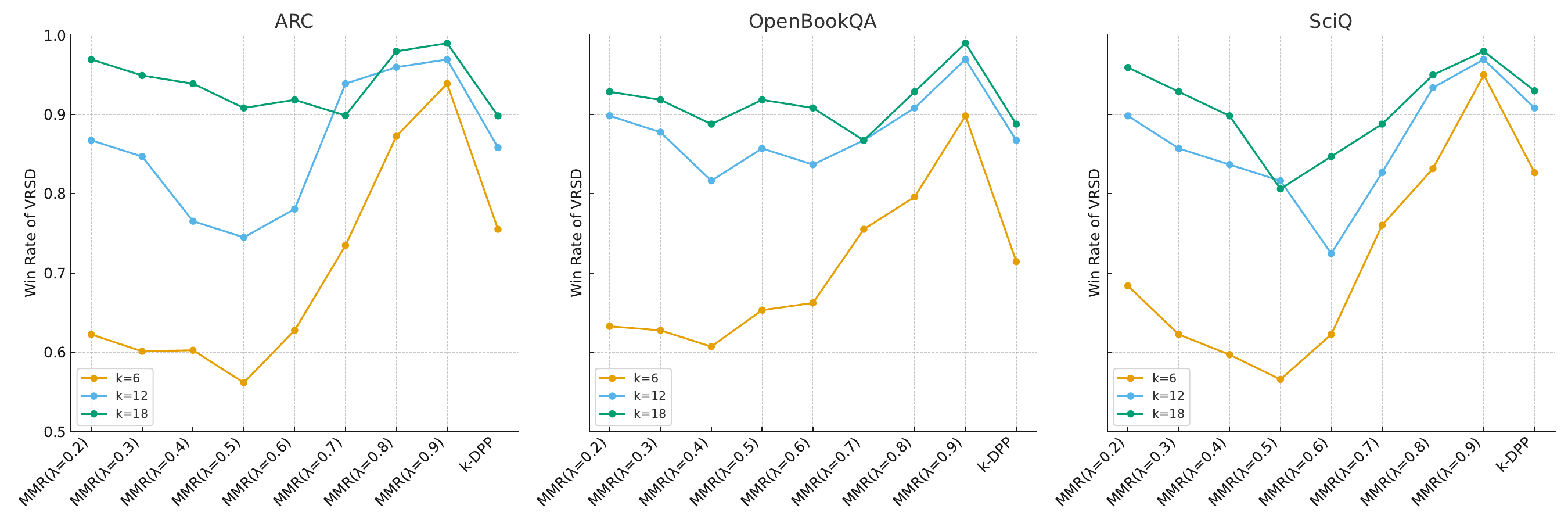}
    \caption{The winning rate of VRSD algorithm against MMR/k-DPP in simulated human evaluations on ARC, OpenBookQA, SciQ datasets under different $\lambda$ values for $k=6$, $k=12$, and $k=18$.}
    \label{fig:subj_eval_1}
    \vspace{-5.5mm}
\end{figure}

We incorporate subjective evaluation using LLMs to simulate human judgment. Recent studies indicate that LLM-based judgments correlate well with human preferences, offering a scalable alternative to traditional evaluation \citep{chiang2023can, dubois2023alpacafarm, chan2023chateval}. Given the strong "world knowledge" of modern LLMs, using them for subjective assessment is both reasonable and effective.

In detail, we employ GPT-4o \citep{openai_gpt4o_2024} to simulate diverse professional personas. Consistent with the objective experiments, we use 20\% of the queries from each dataset. For each query, GPT-4o reviews the original text of the $k$ retrieved results from VRSD and the baselines (MMR/k-DPP). The model is instructed to assign a score based on both relevance and diversity, where a higher score indicates better quality.

To ensure fairness and realism, we constructed a panel of 100 professional roles (e.g., scientists, educators). The detailed prompt and list of personas are provided in the Appendix \ref{detail_SE}. We compute the average score from all evaluators for each query and determine the "Win Rate" of VRSD, defined as the percentage of queries where VRSD outperforms the baselines.

Figure \ref{fig:subj_eval_1} presents results across three datasets. VRSD consistently achieves a win rate exceeding 50\% against MMR regardless of $\lambda$, though margin narrows at $\lambda \approx 0.5$ (consistent with prior findings on MMR's optimal range). Similarly, VRSD maintains a win rate >50\% against k-DPP. We observe that k-DPP performs comparably to MMR with $\lambda \in [0.6, 0.7]$ but lags behind MMR at $\lambda=0.5$.

Notably, VRSD's advantage over both baselines improves steadily as $k$ increases. This reflects a key theoretical distinction: while MMR and k-DPP rely on heuristic repulsion (penalizing similarity), VRSD selects \textit{complementary} items via vector addition. Consequently, as $k$ grows, VRSD more effectively accumulates diverse yet relevant features. These findings align with our objective evaluations, confirming that VRSD consistently outperforms baselines in balancing similarity and diversity. Comprehensive experimental results and ablation analyses are detailed in the Appendix \ref{more}\&\ref{comprehensive_AB}.

\section{Conclusion}

We introduce VRSD, a novel vector retrieval framework that jointly optimizes similarity and diversity through query-to-sum vector alignment. This formulation removes the need for manual parameter tuning and provides an inherent balance between similarity and diversity. By formally defining the resulting combinatorial optimization problem and proving its NP-completeness, we establish a rigorous theoretical foundation for the inherent difficulty of this task. Experiments on multiple benchmarks show that our efficient heuristic consistently outperforms MMR and k-DPP. Overall, VRSD offers a principled and practical alternative for unified retrieval, with geometric insights that inform future retrieval optimization for large language models and in-context learning.

\bibliographystyle{unsrtnat}
\bibliography{custom}

\clearpage
\onecolumn
\tableofcontents

\appendix
\section*{APPENDIX}
\section{Limitations}
The performance of VRSD is inherently based on the geometric properties of the underlying embedding space. Although we verified its robustness across three distinct models, the algorithm assumes that semantic compositionality, specifically through vector addition, holds effectively within the space. Consequently, highly anisotropic or poorly calibrated embedding spaces might limit the precision of the sum vector representation. Furthermore, while our current scope focuses on text-based interactions, real-world retrieval increasingly involves diverse data types. Future work will explore extending the VRSD framework to multimodal information (e.g., images or audio), investigating whether the sum-vector alignment principle generalizes to provide richer representations in multimodal embedding spaces.

\section{Extended Related Work}
Vector retrieval has become a cornerstone technique in a wide range of applications, including search, recommendation, in-context learning, and multimodal matching. The ability to retrieve semantically relevant results using vector similarity has driven progress in tasks such as open domain QA~\citep{chen2017reading,guu2020retrieval}, dense passage retrieval~\citep{khattab2020colbert,ding2020rocketqa}, and retrieval-augmented generation (RAG)~\citep{lewis2020retrieval,mao2020generation}. In addition to open-domain retrieval, vector-based methods have been successfully used in applications such as machine translation~\citep{khandelwal2020nearest}, language modeling~\citep{khandelwal2019generalization,alon2022neuro}, and in-context learning with LLMs~\citep{brown2020language}. Early work relied on sparse retrieval techniques such as BM25~\citep{robertson2009probabilistic} for prompt selection~\citep{liu2021makes}, while more recent approaches employ dense encoders~\citep{reimers2019sentence,shin2021constrained} and contrastive retrievers~\citep{rubin2021learning,wang2023learning}. Efforts such as UPRISE~\citep{cheng2023uprise} and PRAC~\citep{nie2022cross} retrieve demonstrations directly from training corpora to optimize in-context performance. Beyond NLP, vector retrieval is also widely used in recommendation systems. Here, similarity ensures the accuracy of the recommendation, but excessive similarity can lead to filter bubbles and limit content exploration. Diversity-aware methods, including MMR \citep{carbonell1998use}, Gumbel-softmax sampling, and graph-based diversification~\citep{deselaers2009jointly,ye2023complementary}, have been developed to mitigate these issues and promote long-tail item exposure. In multimodal and cross-modal retrieval, such as image-text matching, maintaining diversity is equally important. Methods based on mutual information maximization and contrastive learning~\citep{cheng2023uprise,nie2022cross} enhance retrieval robustness by reducing redundancy and encouraging semantic variety in modalities. Despite these advancements, most existing methods treat each candidate independently during retrieval, neglecting the interactions between selected items. This often results in limited coverage or redundant outputs, particularly in top-$k$ retrieval tasks. MMR remains the most widely used and important algorithm for balancing relevance and diversity, but its reliance on a tunable trade-off parameter $\lambda$ complicates deployment and introduces sensitivity problems~\citep{deselaers2009jointly}. 

In contrast, we adopt an approach fundamentally different from MMR and related algorithms. Instead of measuring relevance and diversity separately, we examine the similarity between the query vector and the sum vector formed by combining multiple candidate vectors, thereby capturing their joint semantics. At first glance, this may appear to be merely an additional similarity constraint in vector retrieval. However, from the perspective of vector space geometry, the similarity between the sum vector and the query vector implicitly introduces a diversity constraint: namely, the individual vectors contributing to the sum must approach the query vector from different directions. This approach effectively leverages the intrinsic properties of vector representations to simultaneously encode both similarity and diversity in retrieval.

\section{Proof of Proposition 1} \label{proof}
The key idea of the VRSD problem is the use of the sum vector to simultaneously capture both similarity and diversity. This holds with the conclusion that the sum of two vectors necessarily lies between the two original vectors. We illustrate this in a two-dimensional space in Figure \ref{fig:fig1}(c). This seems to be a geometric intuition. However, does the sum of two vectors still lie between the two vectors in an n-dimensional space? Moreover, does examining the similarity between the sum vector and the query vector impose both similarity and diversity constraints at the same time?

To address this, we provide an additional justification from both algebraic and algorithmic perspectives below.

\subsection{Algebraic Perspective}

\textbf{Proposition 1}
In a $n$-dimensional vector space, the sum of two vectors necessarily lies between the two original vectors. That is, let $u, v \in \mathbb{R}^n$, and $s=u+v$, so $s$ must be between $u$ and $v$.

\textbf{Proof:} The idea of proof is as follows.

In order to prove this proposition, we first prove that $u$, $v$, and $s$ lie in the same hyperplane, and then prove that the angle between $s$ and $u$ is less than the angle between $u$ and $v$, and the angle between $s$ and $v$ is also less than the angle between $u$ and $v$, therefore, $s$ must be between $u$ and $v$.

Firstly, according to the principles of linear algebra, two linearly independent vectors $u$ and $v$ in a vector space span a two-dimensional subspace (or hyperplane), denoted $\text{span}\{u, v\} \subseteq \mathbb{R}^n$. Any linear combination of $u$ and $v$ necessarily lies within this subspace. Since $s = u + v$ is a specific linear combination of the two, it follows that $s \in \text{span}\{u, v\}$, which means that $s$, $u$, and $v$ all lie in the same hyperplane.

Secondly, we will prove the following conclusion.

Given vector $u$, $v$, and $s=u+v$, let \( \theta \in (0, \pi) \) be the angle between $u$ and $v$, $\alpha$ be the angle between $s$ and $u$, and $\beta$ be the angle between $s$ and $v$, then $\alpha < \theta$ and $\beta < \theta$.

Let $\|u\| = a$, $\|v\| = b$. Then
\[
u \cdot v = \|u\|\|v\|\cos\theta = ab \cos\theta.
\]
and:
\[
\|s\| = \|u + v\| = \sqrt{a^2 + b^2 + 2ab\cos\theta}.
\]

Now consider the angle $\alpha$ between $s$ and $u$, using the cosine formula:
\[
\begin{split}
\cos\alpha &= \frac{s \cdot u}{\|s\| \cdot \|u\|} = \frac{(u + v) \cdot u}{\|s\| \cdot \|u\|} \\
&= \frac{a^2 + ab\cos\theta}{a \cdot \sqrt{a^2 + b^2 + 2ab\cos\theta}}.
\end{split}
\]
Simplify the equation:
\[
\cos\alpha = \frac{a + b\cos\theta}{\sqrt{a^2 + b^2 + 2ab\cos\theta}}.
\]

We now show that \( \cos\alpha > \cos\theta \) for all \( \theta \in (0, \pi) \), which implies \( \alpha < \theta \), because cosine is strictly decreasing on \( (0, \pi) \).

That is,
\[
\cos\alpha = \frac{a + b\cos\theta}{\sqrt{a^2 + b^2 + 2ab\cos\theta}} > \cos\theta.
\]

Multiply both sides by \( \sqrt{a^2 + b^2 + 2ab\cos\theta} \) and simplify the inequality:
\[
a + b\cos\theta > \cos\theta \cdot \sqrt{a^2 + b^2 + 2ab\cos\theta}.
\]

Now we discuss two cases; one is $\theta \in (0, \pi/2]$, the other is $\theta \in (\pi/2, \pi)$.

(1) $\theta \in (0, \pi/2]$.

In this case, $0\leq\cos\theta<1$, 

Squaring both sides (note: all terms are positive for $\theta \in (0, \pi/2]$):
\[
(a + b\cos\theta)^2 > \cos^2\theta (a^2 + b^2 + 2ab\cos\theta).
\]

Expand both sides of the inequality:
\[
\begin{split}
&a^2 + 2ab\cos\theta + b^2\cos\theta^2 \\
&> a^2\cos\theta^2 + b^2\cos\theta^2 + 2ab\cos\theta^3\\
&\Rightarrow a^2 + 2ab\cos\theta > a^2\cos\theta^2 + 2ab\cos\theta^3\\
&\Rightarrow a^2 + 2ab\cos\theta - a^2\cos\theta^2 - 2ab\cos\theta^3 > 0\\
&\Rightarrow a^2(1 - \cos\theta^2) + 2ab\cos\theta(1 - \cos\theta^2) > 0\\
&\Rightarrow (a^2 + 2ab\cos\theta)(1 - \cos\theta^2) > 0\\
\end{split}
\]
The last inequality above is obviously true because $a>0$, $b>0$, and $0\leq\cos\theta<1$. If we work backward from the last inequality, we can get $\cos\alpha>\cos\theta$.

(2)$\theta \in (\pi/2, \pi)$.

In this case, $-1<\cos\theta<0$, to prove
\[
a + b\cos\theta > \cos\theta \cdot \sqrt{a^2 + b^2 + 2ab\cos\theta},
\]

we further divide this case into two subcases.

(2.1) $a + b\cos\theta \geq 0$.

Because $\cos\theta \cdot \sqrt{a^2 + b^2 + 2ab\cos\theta} <0$, we get
\[
a + b\cos\theta > \cos\theta \cdot \sqrt{a^2 + b^2 + 2ab\cos\theta}.
\]

(2.2) $a + b\cos\theta < 0$.

Consider the inequality
\[
a + b\cos\theta > \cos\theta \cdot \sqrt{a^2 + b^2 + 2ab\cos\theta}.
\]

Because both sides of the inequality are less than 0, squaring both sides will result in the inequality being reversed, that is:
\[
\begin{split}
&(a + b\cos\theta)^2 < \cos^2\theta (a^2 + b^2 + 2ab\cos\theta)\\
&\Rightarrow a^2 + 2ab\cos\theta < a^2\cos\theta^2 + 2ab\cos\theta^3\\
&\Rightarrow a^2 + 2ab\cos\theta - a^2\cos\theta^2 - 2ab\cos\theta^3 < 0\\
&\Rightarrow a^2(1 - \cos\theta^2) + 2ab\cos\theta(1 - \cos\theta^2) < 0\\
&\Rightarrow (a^2 + 2ab\cos\theta)(1 - \cos\theta^2) < 0\\
&\Rightarrow a(a + b\cos\theta + b\cos\theta)(1 - \cos\theta^2) < 0
\end{split}
\]

In the above inequality, $a>0$, $(1 - \cos\theta^2)>0$, $a + b\cos\theta < 0$, $b\cos\theta<0$, so $(a + b\cos\theta + b\cos\theta) < 0$. 
Therefore, the last inequality above is also true, and if we work backward, we can get $\cos\alpha>\cos\theta$.

Thus,
\[
\cos\alpha > \cos\theta \Rightarrow \alpha < \theta.
\]

By symmetry, the same result applies to the angle $\beta$ between $s$ and $v$, that is,
\[
\cos\beta > \cos\theta \Rightarrow \beta < \theta.
\]

In summary, $u$, $v$, and $s$ lie in the same hyperplane, the angle between $s$ and $u$ is less than the angle between $u$ and $v$, and the angle between $s$ and $v$ is also less than the angle between $u$ and $v$, therefore, $s$ must be between $u$ and $v$.
\hfill\qedsymbol

\subsection{Algorithmic Perspective:}
Consider a query vector $q$ and a set of candidate vectors $d_0, d_1, \ldots, d_n$, sorted by their similarity to $q$. When the goal is to select a subset of vectors such that the sum vector is maximally similar to $q$, the algorithm first selects $d_0$ (since the sum of a single vector is itself). For the second selection, the next best candidate $d_x$ is chosen such that the new sum $s = d_0 + d_x$ is as close as possible to $q$. Based on the \textbf{Proposition 1}, this sum vector $s$ lies between $d_0$ and $d_x$, which implies that $d_0$ and $d_x$ are positioned on either side of $s$, and hence $d_0$ and $d_x$ are positioned on either side of $q$ when optimized. This introduces an implicit diversity constraint, since selected vectors are encouraged to approach $q$ from different directions.

Repeating this process means that each new vector is selected so that its addition pushes the sum vector closer to the query while maintaining angular spread, achieving a balance between similarity and diversity. Although our method superficially appears to only consider the similarity between the sum vector and the query, this selection mechanism inherently incorporates a diversity objective, which is also empirically validated in our experiments.

\section{Details on Baselines and Datasets}
\label{app:setup_details}
To ensure a rigorous evaluation, we benchmark VRSD against strong, industry-standard baselines and evaluate performance across diverse scientific question-answering datasets. This section details the theoretical formulation of the baselines, the implementation details, and the characteristics of the datasets used.

\subsection{Baselines}

We compare VRSD against two widely used methods that represent distinct paradigms (heuristic vs. probabilistic) in relevance-diversity trade-offs.

\paragraph{Maximal Marginal Relevance (MMR)}
MMR \citep{carbonell1998use} serves as our primary heuristic baseline. It greedily selects items by explicitly penalizing the similarity between a candidate item and the set of already selected items. Mathematically, it trades off query relevance and novelty via a weighting parameter $\lambda$. 
Due to its effectiveness and simplicity, MMR has become a standard reranking heuristic in Information Retrieval (IR) and Retrieval-Augmented Generation (RAG) pipelines. It is implemented out-of-the-box in popular frameworks such as LangChain \citep{chase2022langchain} and LlamaIndex \citep{Liu_LlamaIndex_2022}, which facilitates faithful and reproducible comparisons in our experiments.

\paragraph{Determinantal Point Processes (k-DPP)}
Complementary to the greedy nature of MMR, Determinantal Point Processes (DPP) provide a probabilistic framework for modeling diversity. A DPP assigns a probability measure to subsets of a finite candidate pool, where the probability of selecting a subset is proportional to the determinant of its kernel matrix. 
Intuitively, a DPP assigns larger scores to subsets whose items are both individually high-quality (high relevance) and mutually different (low pairwise redundancy), effectively implementing a built-in notion of ``repulsion'' between items.

In our experiments, we utilize the fixed-size variant, \textbf{k-DPP} \citep{kulesza2011kdpp}, which conditions the model to return exactly $k$ items. This aligns with common top-$k$ retrieval settings. To ensure scalability, we employ the \textbf{Fast-MAP-DPP} algorithm \citep{chen2018fastmap}, which makes DPP inference practical at the reranking scale without sacrificing relevance \citep{kulesza2012dpp, gillenwater2012nearoptimal, li2016efficientkdpp}. 
k-DPP has been successfully applied to diversify recommendation lists, improve multi-document summarization, and perform supervised video summarization \citep{gong2014seqdpp, sharghi2018seqdppimprove, cho2019mdsdpp}. Given these properties, k-DPP serves as a robust probabilistic baseline.

\subsection{Embedding Model}
For the main experiments, we employ the \texttt{all-mpnet-base-v2} embedding model \citep{song2020mpnet}. This is a robust, MPNet-based general-purpose encoder widely recognized as a strong baseline for dense retrieval and RAG systems \citep{all_mpnet_base_v2_card}. It provides a stable semantic space for evaluating the geometric properties of the sum vector. 

\subsection{Dataset Descriptions}
We utilize three publicly available datasets focusing on scientific and commonsense reasoning to evaluate the retrieval performance.

\begin{itemize}
    \item \textbf{ARC-DA} \citep{bhakthavatsalam2021think}: The ARC-DA (AI2 Reasoning Challenge, Direct Answer) dataset is designed to evaluate question answering systems on science-related multiple-choice questions. It requires the retrieval system to identify precise scientific facts to support reasoning.
    
    \item \textbf{OpenBookQA} \citep{mihaylov2018openbookqa}: This dataset is constructed to test a system’s ability to combine scientific facts with common sense knowledge. It consists of elementary-level science questions that require multi-hop reasoning, making diversity in retrieval crucial for covering different aspects of the required knowledge.
    
    \item \textbf{SciQ} \citep{welbl2017sciq}: The SciQ dataset is a crowdsourced collection of multiple-choice science questions derived from real science curricula. It is intended to support the development and evaluation of educational QA systems, providing a testbed for retrieving academically relevant explanations.
\end{itemize}

\section{More Experimental Results} \label{more}
\begin{figure*}[htbp]
    \centering
    \begin{subfigure}{\textwidth}
        \includegraphics[width=\textwidth, height=0.25\textheight]{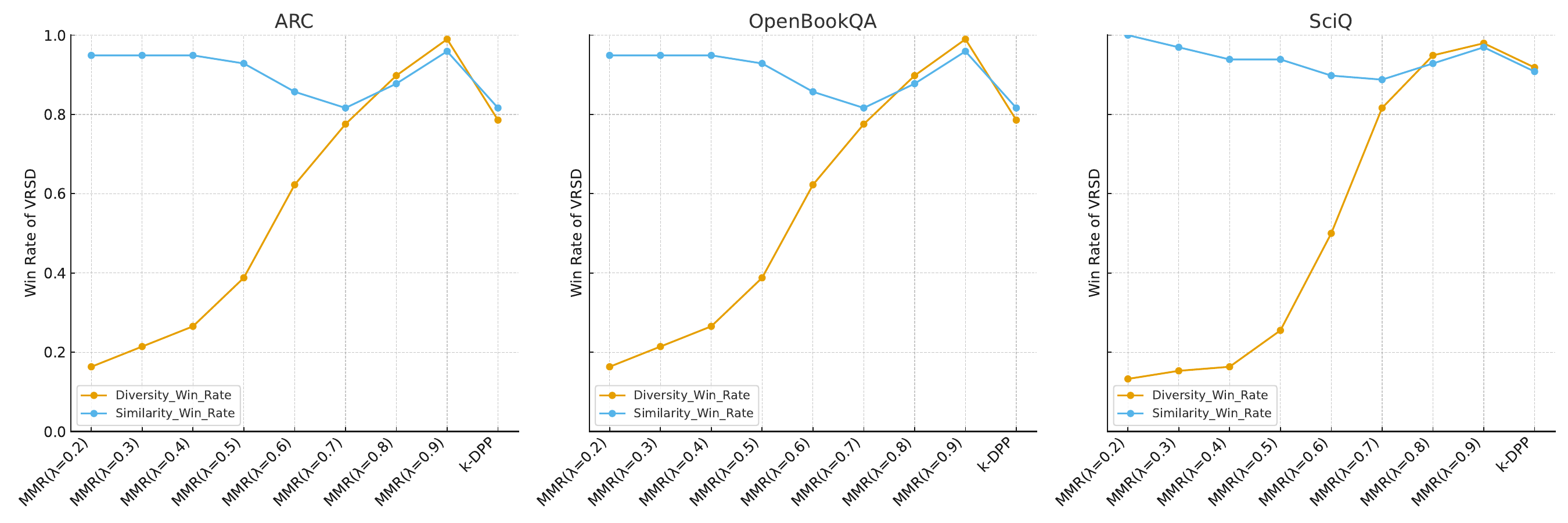}
        \caption{k=6}
    \end{subfigure}
    \begin{subfigure}{\textwidth}
        \includegraphics[width=\textwidth, height=0.25\textheight]{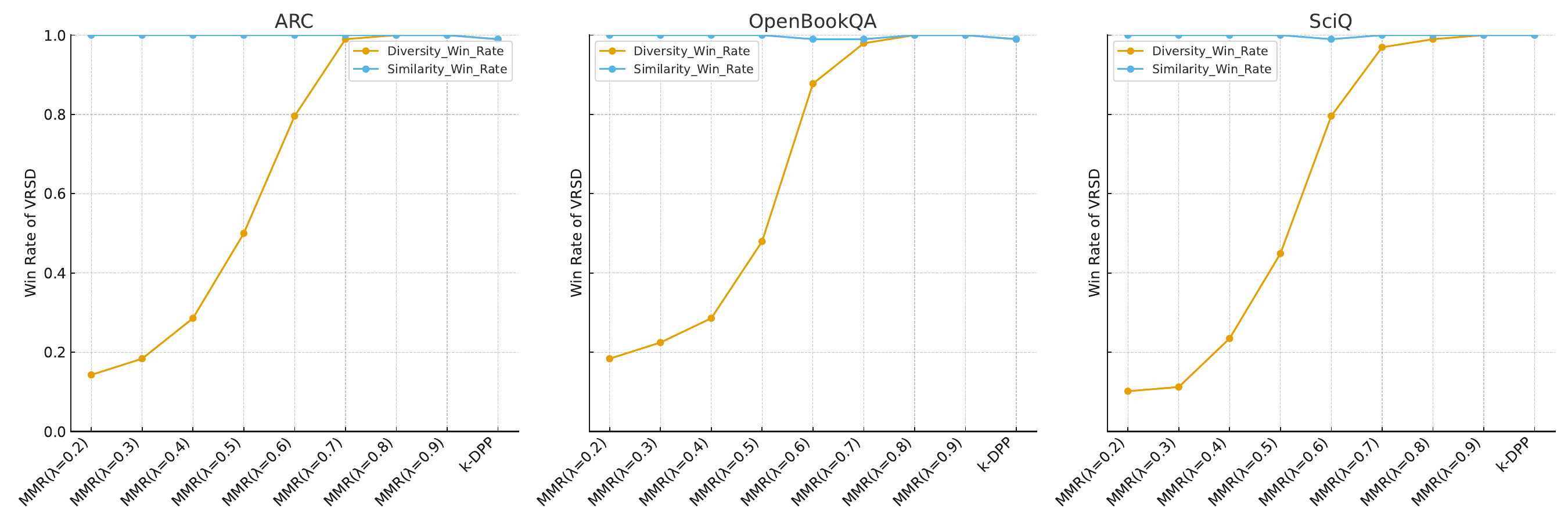}
        \caption{k=12}
    \end{subfigure}
    \begin{subfigure}{\textwidth}
        \includegraphics[width=\textwidth, height=0.25\textheight]{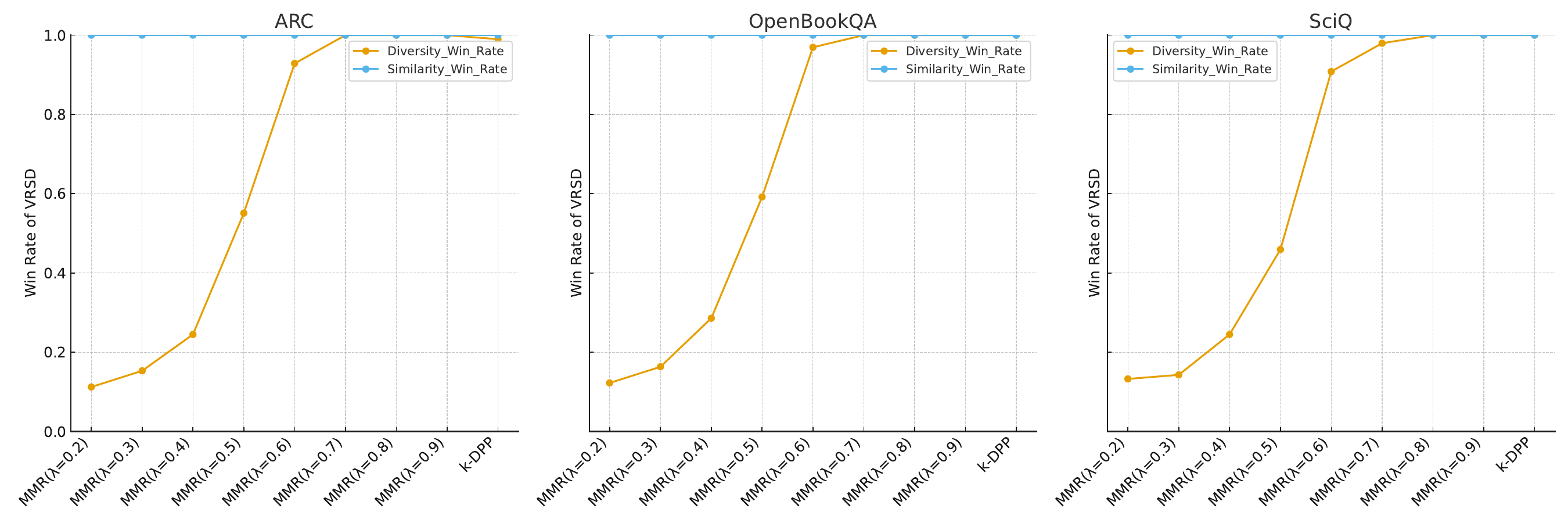}
        \caption{k=18}
    \end{subfigure}
    \caption{Win Rate of VRSD vs MMR/k-DPP across Different Datasets (k=6,12,18).}
    \label{fig:winrate_all}
\end{figure*}

\subsection{Experimental Objectives}
Our experiments aim to address two primary objectives.

\begin{itemize}
    \item In the formulation of the VRSD problem, we use the similarity between the sum vector and the query vector to capture both similarity and diversity constraints in vector retrieval. A key question, therefore, is whether this similarity measure - between the sum vector and the query vector - can effectively satisfy the requirements of both similarity and diversity in real-world retrieval tasks. This can be empirically tested by examining the performance of the VRSD algorithm in terms of similarity and diversity. Although the VRSD is a heuristic algorithm and does not guarantee an optimal solution to the VRSD problem, its objective is to maximize the similarity between the sum vector and the query vector, which should inherently satisfy the similarity requirement in retrieval. On the other hand, as previously discussed, the requirement to maximize this similarity implicitly encourages the constituent vectors of the sum vector to approach the query vector from diverse directions, thereby embedding a diversity constraint. Therefore, if the VRSD algorithm also exhibits strong performance in promoting diversity during retrieval, it would provide further evidence that using the similarity between the sum vector and the query vector is a reasonable and effective formulation for addressing both constraints in the VRSD problem.

    \item Conduct a thorough comparison between the VRSD algorithm and the MMR/k-DPP algorithm in terms of both similarity and diversity. Both VRSD and MMR/k-DPP are approaches that aim to balance similarity and diversity in vector retrieval. However, they adopt fundamentally different strategies. MMR performs this balance by adjustment of a parameter $\lambda$, which requires empirical tuning and lacks a principled way to determine its optimal value in advance. k-DPP provides a probabilistic way to favor diverse yet high-quality subsets from a finite candidate pool. In contrast, the VRSD algorithm achieves this balance by explicitly maximizing the similarity between the sum vector and the query vector, thus directly enforcing similarity constraints and implicitly introducing diversity constraints. This makes VRSD a more adaptive method that inherently balances the two objectives. Which of these two approaches is more effective?
\end{itemize}

To verify these two objectives, we designed both objective and subjective experiments. Objective experiments involve evaluating the results of VRSD and MMR/k-DPP using similarity metrics and diversity metrics. Subjective experiments simulate human evaluation using LLM to mimic ordinary users to score the retrieval results of both algorithms. Given the powerful "world knowledge" of modern LLM, this simulation-based subjective evaluation is reasonable and credible.

The following provides additional objective experimental results and further details of the subjective evaluation.

\begin{table*}[ht]
\centering
\caption{Prompt template used for LLM-based evaluation.}
\resizebox{\textwidth}{!}{%
\begin{tabular}{p{3cm} | p{10.5cm}}
\hline
\textbf{Field} & \textbf{Content / Template} \\
\hline
System message &
\texttt{You are an evaluation expert. Please score strictly as '\{role\}'.}
\\
\hline
User prompt &
\texttt{You need to evaluate the following retrieval results from the perspective of '\{role\}', strictly following the instructions below.}\\
 & \texttt{Query: \{query\}}\\
 & \texttt{Retrieved Results:}\\
 & \texttt{- \{item\_1\}}\\
 & \texttt{- \{item\_2\}}\\
 & \texttt{...}\\
 & \texttt{\{eval\_instruction\}}\\
 & \texttt{Please strictly output the result in the following JSON format (do not add extra explanations):}\\
 & \texttt{\{"role": "<role name>", "score": <integer score>, "comment": "<brief reason>"\}}
\\
\hline
Evaluation instruction & \texttt{The retrieval results are obtained from the database using similarity and diversity retrieval algorithms. Please comprehensively evaluate the similarity and diversity of these retrieval results with respect to the query, according to the similarity and diversity criteria. Give your evaluation as a score (1-100), and briefly explain the reason in one sentence. A higher score indicates better similarity and diversity with respect to the query; a lower score indicates worse.}
\\
\hline
Variables & \texttt{\{role\}}, \texttt{\{query\}}, \texttt{\{item\_i\}}, \texttt{\{eval\_instruction\}} are dynamically filled in each evaluation instance.
\\
\hline
\end{tabular}
}
\label{tab:llm-prompt-template}
\end{table*}

\subsection{Additional Objective Experiments: Microscopic Comparison of VRSD and MMR/k-DPP}

While Table1 in the paper report average performance, we further analyze individual retrieval results. 

For each query result obtained by VRSD and MMR/k-DPP (with $\lambda$ values ranging from 0.2 to 0.9), we calculate their similarity and diversity scores. Then, aggregate all retrieval results to calculate the percentage of cases in which VRSD outperforms MMR/k-DPP in terms of similarity, which is referred to as the similarity win rate of VRSD over MMR/k-DPP, and the percentage of cases in which VRSD outperforms MMR/k-DPP in terms of diversity, which is referred to as the diversity win rate of VRSD over MMR/k-DPP. That is,

\begin{equation}
\label{eq:sim-win-rate}
\mathrm{Similarity\ Win\ Rate}
=
\frac{
\mathrm{count}\left(
\mathrm{Similarity}_{\mathrm{VRSD}}
>
\mathrm{Similarity}_{\mathrm{MMR/k\text{-}DPP}}
\right)
}{
\mathrm{Total\ number\ of\ queries}
}.
\end{equation}

\begin{equation}
\label{eq:div-win-rate}
\mathrm{Diversity\ Win\ Rate}
=
\frac{
\mathrm{count}\left(
\mathrm{Diversity}_{\mathrm{VRSD}}
<
\mathrm{Diversity}_{\mathrm{MMR/k\text{-}DPP}}
\right)
}{
\mathrm{Total\ number\ of\ queries}
}.
\end{equation}

Figure \ref{fig:winrate_all} in this Appendix illustrates VRSD’s win rate vs. MMR/k-DPP in terms of similarity and diversity. Figure \ref{fig:winrate_all}(a) shows the situation when $k=6$, Figure \ref{fig:winrate_all}(b) for $k=12$, and Figure \ref{fig:winrate_all}(C) for $k=18$. It can also be seen that VRSD maintains a consistent advantage in similarity, with similarity win rates often exceeding 80\%, even close to 1. For diversity, MMR performs better when $\lambda < 0.5$, but VRSD gradually overtakes $\lambda > 0.5$ and continues to improve.

From Figure \ref{fig:winrate_all}, it can also be seen that as $k$ increases, the advantage of VRSD diversity over MMR/k-DPP becomes more pronounced. 

\subsection{Details of the Subjective Experiments}
\label{detail_SE}
In detail, we employ GPT-4o \citep{openai_gpt4o_2024} to simulate people from diverse professional backgrounds in scoring the retrieval results of VRSD and MMR. The detailed procedure is as follows.

To ensure consistency with previous experiments, we again use 20\% of each data set as the retrieval query set and submit each query to both the VRSD and MMR algorithms. For each query, both algorithms retrieve $k$ results. Unlike the objective evaluation, we retain the original textual content of the retrieved results rather than their vector representations.

We then present the query text and the $k$ retrieved texts to GPT-4o, assigning it the persona of a specific professional. The GPT-4o is instructed to rate the retrieval results based on both similarity and diversity with respect to the query. A higher score indicates that the $k$ results are both highly relevant and sufficiently diverse.

Table~\ref{tab:llm-prompt-template} is the prompt template used for the LLM-based evaluation.

\subsubsection{Specification of Simulated Evaluator Roles}
To ensure fairness, realism, and domain diversity in the simulated evaluation process, we constructed a panel of 100 professional roles guided by the following principles:

\begin{itemize}
    \item \textbf{Comprehensive Coverage of Industry Domains:} The set of professionals spans many broad sectors, including healthcare, education, technology, finance, law, science and engineering, social services, agriculture, and others. This wide coverage reflects a realistic cross section of modern professional society.

    \item \textbf{Balanced Representation to Prevent Domain Dominance:} Each sector contributes only a limited number of roles to avoid over representation of commonly studied or high-profile professions such as “manager” or “engineer”. This promotes equitable distribution across diverse knowledge domains.

    \item \textbf{Diversity of Professional Functions:} The selected roles encompass a wide spectrum of occupations-from academic and analytical to hands-on and service-oriented-ensuring the inclusion of technical, managerial, and operational perspectives.

    \item \textbf{Neutrality with Respect to Demographic or Socioeconomic Bias:} The role selection avoids professions that may carry strong gender, geographic or hierarchical connotations, thus aiming at an impartial simulation of professional perspectives.
\end{itemize}
These design choices collectively contribute to a more representative and equitable simulated evaluation framework, mitigating potential biases, and improving the generalization of the results.

The following is the list of 100 simulated professional evaluators.

\subsubsection{List of 100 Simulated Professional Evaluators}
\begin{quote}
1. Education: High School History Teacher\\
2. Healthcare: Cardiologist\\
3. Law: Corporate Legal Advisor\\
4. Technology: AI Researcher\\
5. Finance: Bank Risk Manager\\
6. Journalism: Senior Reporter\\
7. Logistics: Warehouse Dispatcher\\
8. Manufacturing: Mechanical Engineer\\
9. Healthcare: Public Health Expert\\
10. Education: University Lecturer (Economics)\\
11. Agriculture: Farm Owner\\
12. IT: Frontend Developer\\
13. Retail: Supermarket Manager\\
14. Food Services: Head Chef\\
15. Public Services: Community Service Worker\\
16. Arts: Graphic Designer\\
17. Research: Physics Lab Researcher\\
18. Gaming: Game Designer\\
19. Film: Documentary Director\\
20. Law: Prosecutor\\
21. Sports: Athletic Trainer\\
22. Architecture: Interior Designer\\
23. Energy: Petroleum Engineer\\
24. Healthcare: Head Nurse\\
25. Technology: Data Scientist\\
26. Education: K12 Programming Teacher\\
27. Agriculture: Agricultural Supply Chain Analyst\\
28. Finance: Securities Analyst\\
29. Business: E-commerce Operations Specialist\\
30. Culture: Book Editor\\
31. Aviation: Aircraft Maintenance Technician\\
32. Healthcare: Pharmacist\\
33. Arts: Freelance Illustrator\\
34. Energy: Wind Farm Manager\\
35. IT: Systems Architect\\
36. Architecture: Structural Engineer\\
37. Media: Content Creator\\
38. Food Services: Barista\\
39. Technology: Autonomous Vehicle Developer\\
40. Education: Education Policy Researcher\\
41. Healthcare: Clinical Data Analyst\\
42. Finance: Financial Advisor\\
43. Research: Environmental Scientist\\
44. Public Affairs: Policy Advisor\\
45. Arts: Art Critic\\
46. Gaming: Game Tester\\
47. Manufacturing: Industrial Robot Operator\\
48. Journalism: Foreign Correspondent\\
49. Finance: Venture Capital Partner\\
50. Public Services: Urban Transportation Planner\\
51. Education: Special Education Teacher\\
52. Logistics: International Logistics Coordinator\\
53. IT: Cybersecurity Analyst\\
54. Research: Anthropology PhD Student\\
55. Healthcare: Nutritionist\\
56. Finance: Actuary\\
57. Education: Study Abroad Consultant\\
58. IT: Database Administrator\\
59. Food Services: Pastry Chef\\
60. Culture: Playwright\\
61. Manufacturing: Semiconductor Process Engineer\\
62. Agriculture: Agricultural Extension Officer\\
63. Arts: Contemporary Art Curator\\
64. Public Services: Sanitation Manager\\
65. Energy: Nuclear Plant Operator\\
66. Technology: Blockchain Developer\\
67. Gaming: Indie Game Developer\\
68. Architecture: BIM Modeler\\
69. Journalism: Tech Section Editor\\
70. Education: Online Course Developer\\
71. IT: DevOps Engineer\\
72. Finance: Cryptocurrency Trader\\
73. Healthcare: Genomics Specialist\\
74. Business: Brand Strategy Consultant\\
75. Law: Environmental Lawyer\\
76. Public Affairs: Policy Research Assistant\\
77. Agriculture: Smart Farming Engineer\\
78. Energy: Solar Power Engineer\\
79. Arts: Photographer\\
80. Education: College Admissions Director\\
81. Technology: Human-Computer Interaction Expert\\
82. Logistics: Express Delivery Manager\\
83. Food Services: F\&B Marketing Specialist\\
84. Manufacturing: Aerospace Equipment Technician\\
85. Journalism: Investigative Journalist\\
86. Law: Family Dispute Mediator\\
87. IT: Technical Product Manager\\
88. Public Services: Fire Safety Evaluator\\
89. Arts: Music Producer\\
90. Healthcare: Rehabilitation Therapist\\
91. Business: Startup Founder\\
92. Finance: Credit Analyst\\
93. Education: Adult Vocational Trainer\\
94. Research: Neuroscientist\\
95. IT: Cloud Architect\\
96. Public Affairs: NGO Project Executive\\
97. Architecture: Sustainable Building Consultant\\
98. Gaming: Narrative Designer\\
99. Culture: Ancient Book Restorer\\
100. Public Safety: Online Public Opinion Analyst\\
\end{quote}

\section{Ablation Study on Embedding Models}
\label{comprehensive_AB}
In contrast to MMR, VRSD does not require parameter tuning. However, because the choice of embedding model can influence vector retrieval, we conducted an ablation study across different embeddings. More specifically, beyond the previously employed all-mpnet-base-v2 embedding model \citep{all_mpnet_base_v2_card, song2020mpnet}, we repeated the main experiments using two additional and distinct embedding models, namely bge-m3 \citep{bge_m3_card, chen2024bge}, and all-MiniLM-L6-v2 \citep{all_minilm_l6_v2_card, wang2020minilm}. Our question is whether VRSD consistently achieves a better relevance--diversity balance across different representation spaces (varying in dimensionality and training paradigm) and whether its advantage grows with the target set size $k$.

\subsection{Embedding Models and Their Characteristics}
\label{app:emb:model}
\noindent\textbf{\texttt{all-mpnet-base-v2}} (768 dims, English): a robust MPNet-based general-purpose encoder that serves as a strong baseline for retrieval/RAG systems. It ensures stable semantic consistency and produces relatively uniform embeddings, which makes it ideal for assessing methods in a mature English semantic space.\\
\textbf{\texttt{bge-m3}} (1024 dims, multilingual, multi-functional): a BGE-family model trained with instruction-style signals. This model often produces clearer semantic clusters across languages and domains, enabling us to observe algorithmic behavior in a higher-dimensional, cluster-friendly space.\\
\textbf{\texttt{all-MiniLM-L6-v2}} (384 dims, English, lightweight): a compact low-latency encoder with lower capacity/resolution. It is used to test algorithmic stability in a lower-dimensional (slightly noisier) embedding space.

\subsection{Main Findings}
\label{app:emb:results}
The experimental setup and evaluation metrics are identical to those described in Section 4.
Table \ref{tab:sim-div-models} summarizes the results on ARC dataset for three encoders and multiple $k$'s, Figure \ref{fig:subj_eval_models} shows the win rate of VRSD versus MMR/k-DPP when considering the balance between relevance and diversity. The following is what can be seen from Table \ref{tab:sim-div-models} and Figure \ref{fig:subj_eval_models}.
\\\textbf{(1) Stable advantage of VRSD.} Across all encoders and values of $k$, VRSD achieves higher similarity than strong baselines while achieving good diversity, producing a consistently better relevance--diversity balance. For example, under all-mpnet-base-v2 and all-MiniLM-L6-v2, VRSD simultaneously improves Sim. Mean (larger is better) and reduces Div. Mean (smaller is better) among the selected items.
\\\textbf{(2) Performance differences between MMR and k-DPP across embedding models.} With all-mpnet-base-v2 and all-MiniLM-L6-v2, MMR (at its near-optimal $\lambda$) tends to slightly outperform $k$-DPP; with \texttt{bge-m3}, $k$-DPP is relatively stronger. This likely reflects \texttt{bge-m3}'s clearer cluster structure and higher dimensionality, where k-DPP's repulsion can more effectively avoid intra-cluster duplication.
\\\textbf{(3) VRSD gains grow with larger $k$.} When $k$ increases from 6 to 18, VRSD's improvements over MMR/$k$-DPP become more pronounced in both higher Similarity and lower Diversity. As the value of $k$ gets larger and the global layout matters more, VRSD better controls redundancy while maintaining strong query relevance.

\subsection{Analysis and Discussion}
\label{app:emb:analysis}
\paragraph{Why is VRSD more stable?}
MMR linearly trades off “pull-to-query” and “repel-among-results” via $\lambda$; $k$-DPP encourages diversity via a determinant-based repulsion that can, in some spaces, push items outward and sacrifice query relevance. Both essentially \emph{separate} relevance and repulsion. In contrast, VRSD centers on \emph{complementarity}: it selects items that simultaneously contribute to the query’s semantics and add non-redundant information to the current set. This unified criterion adapts better across embedding geometries and, crucially, scales with $k$: as $k$ grows, complementarity keeps bringing in new facets instead of repeating existing ones, improving both relevance and diversity.

\paragraph{Conclusion.}
VRSD exhibits more stable and robust relevance--diversity trade-offs across three substantially different embedding models and multiple $k$ values, with its advantage amplifying as $k$ increases. This supports our thesis that \emph{complementarity} (based on the sum vector), as a unified constraint, adapts better across vector spaces than strategies that separately optimize relevance and repulsion (MMR/$k$-DPP).

\begin{table*}[htb]
\centering
\caption{Similarity and Diversity scores of different algorithms on ARC, using three different embedding models.}
\resizebox{\textwidth}{!}{%
\begin{tabular}{l|l|cc|cc|cc}
\toprule
$k$ & \textbf{Algorithms} & \multicolumn{2}{c|}{\textbf{ARC(all-mpnet-base-v2)}} & \multicolumn{2}{c|}{\textbf{ARC(bge-m3)}} & \multicolumn{2}{c}{\textbf{ARC(all-MiniLM-L6-v2)}} \\
 &  & Sim. Mean & Div. Mean & Sim. Mean & Div. Mean & Sim. Mean & Div. Mean \\
\midrule
\multirow{9}{*}{6}
& MMR($\lambda$=0.2) & 0.6797 & 0.2766 & 0.7404 & 0.4676 & 0.6659 & 0.2685 \\
& MMR($\lambda$=0.3) & 0.6859 & 0.2789 & 0.7427 & 0.4688 & 0.6695 & 0.2720 \\
& MMR($\lambda$=0.4) & 0.6913 & 0.2844 & 0.7462 & 0.4719 & 0.6748 & 0.2813 \\
& \textbf{MMR($\lambda$=0.5)} & \textbf{0.6987} & \textbf{0.2988} & \textbf{0.7504} & \textbf{0.4773} & \textbf{0.6840} & \textbf{0.2919} \\
& \textbf{MMR($\lambda$=0.6)} & \textbf{0.7047} & \textbf{0.3178} & \textbf{0.7542} & \textbf{0.4912} & \textbf{0.6922} & \textbf{0.3106} \\
& MMR($\lambda$=0.7) & 0.7065 & 0.3396 & 0.7558 & 0.5050 & 0.6937 & 0.3323 \\
& MMR($\lambda$=0.8) & 0.7022 & 0.3631 & 0.7536 & 0.5216 & 0.6893 & 0.3583 \\
& MMR($\lambda$=0.9) & 0.6928 & 0.3914 & 0.7470 & 0.5438 & 0.6776 & 0.3941 \\
& k-DPP & 0.7081 & 0.3382 & 0.7585 & 0.4838 & 0.6932 & 0.3386 \\
& \textbf{VRSD} & \textbf{0.7161} & \textbf{0.3109} & \textbf{0.7614} & \textbf{0.4878} & \textbf{0.7052} & \textbf{0.2976} \\
\midrule
\multirow{9}{*}{12}
& MMR($\lambda$=0.2) & 0.6780 & 0.2399 & 0.7428 & 0.4518 & 0.6647 & 0.2348 \\
& MMR($\lambda$=0.3) & 0.6869 & 0.2442 & 0.7458 & 0.4533 & 0.6710 & 0.2384 \\
& MMR($\lambda$=0.4) & 0.6981 & 0.2535 & 0.7505 & 0.4569 & 0.6797 & 0.2448 \\
& \textbf{MMR($\lambda$=0.5)} & \textbf{0.7073} & \textbf{0.2661} & \textbf{0.7548} & \textbf{0.4642} & \textbf{0.6935} & \textbf{0.2580} \\
& \textbf{MMR($\lambda$=0.6)} & \textbf{0.7147} & \textbf{0.2857} & \textbf{0.7589} & \textbf{0.4753} & \textbf{0.7005} & \textbf{0.2805} \\
& MMR($\lambda$=0.7) & 0.7143 & 0.3122 & 0.7609 & 0.4881 & 0.7000 & 0.3057 \\
& MMR($\lambda$=0.8) & 0.7100 & 0.3323 & 0.7579 & 0.5058 & 0.6936 & 0.3312 \\
& MMR($\lambda$=0.9) & 0.7022 & 0.3506 & 0.7536 & 0.5191 & 0.6841 & 0.3523 \\
& k-DPP & 0.7155 & 0.3136 & 0.7634 & 0.4668 & 0.7003 & 0.3061 \\
& \textbf{VRSD} & \textbf{0.7332} & \textbf{0.2695} & \textbf{0.7682} & \textbf{0.4638} & \textbf{0.7203} & \textbf{0.2607} \\
\midrule
\multirow{9}{*}{18}
& MMR($\lambda$=0.2) & 0.6712 & 0.2213 & 0.7381 & 0.4411 & 0.6572 & 0.2139 \\
& MMR($\lambda$=0.3) & 0.6789 & 0.2253 & 0.7418 & 0.4433 & 0.6662 & 0.2164 \\
& MMR($\lambda$=0.4) & 0.6918 & 0.2332 & 0.7467 & 0.4481 & 0.6747 & 0.2248 \\
& \textbf{MMR($\lambda$=0.5)} & \textbf{0.7039} & \textbf{0.2475} & \textbf{0.7522} & \textbf{0.4547} & \textbf{0.6901} & \textbf{0.2378} \\
& \textbf{MMR($\lambda$=0.6)} & \textbf{0.7117} & \textbf{0.2674} & \textbf{0.7567} & \textbf{0.4635} & \textbf{0.6983} & \textbf{0.2584} \\
& MMR($\lambda$=0.7) & 0.7105 & 0.2900 & 0.7576 & 0.4776 & 0.6975 & 0.2826 \\
& MMR($\lambda$=0.8) & 0.7035 & 0.3143 & 0.7552 & 0.4928 & 0.6911 & 0.3055 \\
& MMR($\lambda$=0.9) & 0.6944 & 0.3326 & 0.7506 & 0.5062 & 0.6800 & 0.3269 \\
& k-DPP & 0.7112 & 0.2933 & 0.7605 & 0.4568 & 0.6969 & 0.2864 \\
& \textbf{VRSD} & \textbf{0.7344} & \textbf{0.2454} & \textbf{0.7666} & \textbf{0.4574} & \textbf{0.7223} & \textbf{0.2358} \\
\bottomrule
\end{tabular}
}
\label{tab:sim-div-models}
\end{table*}

\begin{figure*}[htb]
    \centering
    \includegraphics[width=\textwidth]{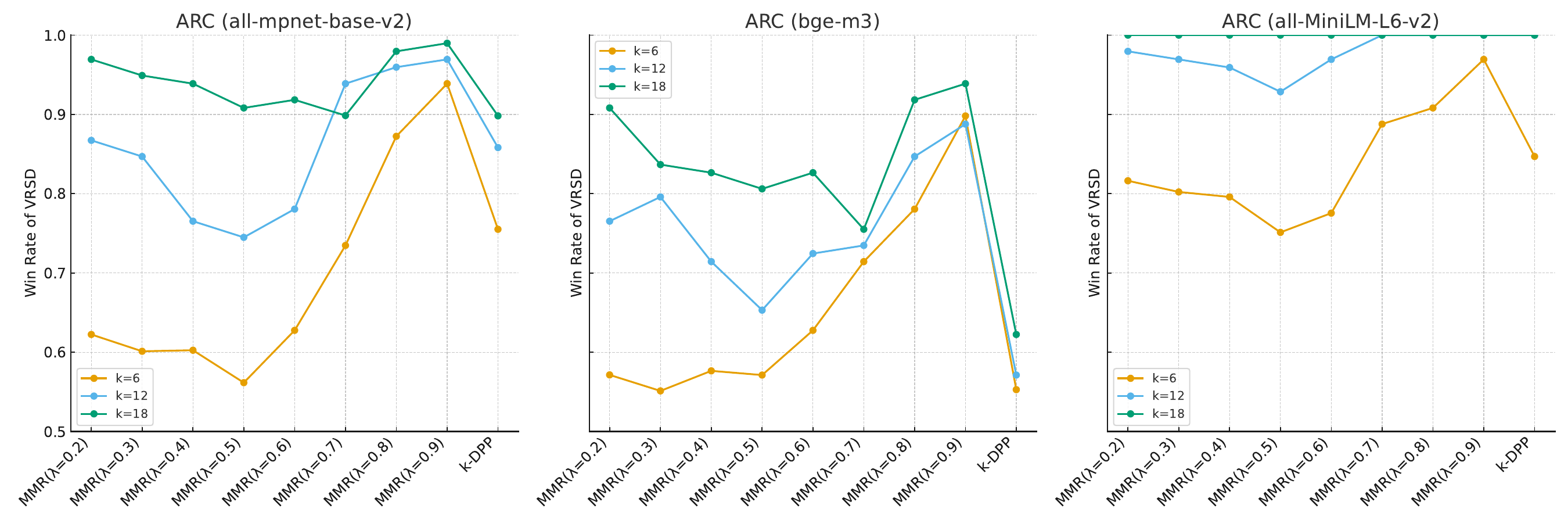}
    \caption{The winning rate of VRSD algorithm against MMR/k-DPP algorithms in simulated human evaluations on ARC dataset using different embedding models.}
    \label{fig:subj_eval_models}
\end{figure*}


\end{document}